\documentclass[a4paper,12pt]{article}

\usepackage[colorlinks]{hyperref}
\usepackage[left=3cm,right=3cm]{geometry}
\usepackage{bbm}
\usepackage{enumerate}
\usepackage{amsmath}
\usepackage{amsfonts}
\usepackage{amssymb}
\usepackage{amsthm}

\newcommand{\Prob}[1]{\ensuremath{\mathbb{P}\left(#1\right)}}
\newcommand{\Exp}[1]{\ensuremath{\mathbb{E}\left(#1\right)}}
\newcommand{\R}{\ensuremath{\mathbb{R}}}
\newcommand{\Oh}[1]{\ensuremath{\textrm{O}\left(#1\right)}}
\newcommand{\oh}[1]{\ensuremath{\textrm{o}\left(#1\right)}}
\newcommand{\mydiff}[1]{\ensuremath{\,{\rm d}#1}}
\newcommand{\abs}[1]{\left\lvert#1\right\rvert}
\newcommand{\fancy}[1]{\ensuremath{\mathcal{#1}}}
\newcommand{\F}{\mathbb{F}}
\newcommand{\curlybr}[1]{\ensuremath{\left\{#1\right\}}}
\renewcommand{\P}{\ensuremath{\mathbb{P}}}

\theoremstyle{plain}
 \newtheorem{Theorem}{Theorem}[section]
 
 \newtheorem{Proposition}{Proposition}[section]
 
 \newtheorem{Lemma}{Lemma}[section]
 
\theoremstyle{definition}
 \newtheorem{Definition}{Definition}[section]

 \newtheorem{Notation}{Notation}[section]
 \newtheorem{Remark}{Remark}[section]

\title{Implied volatility explosions: European calls and implied volatilities close to expiry in exponential L{\'e}vy models}
\author{Michael Roper\thanks{School of Mathematics and Statistics, University of New South Wales;  email:s3115607@science.unsw.edu.au}}

\begin{document}
\maketitle

\begin{abstract}
We examine the small expiry behaviour of the price of call options in models of exponential L{\' e}vy type. In most cases of interest, it turns out that
\begin{equation*}
	\Exp{(S_\tau-K)^+} - (S_0-K)^+ \sim 	\begin{cases}
							\tau\int_\R(S_0e^x-K)^+\,\nu(\textrm{d}x),&S_0<K,\\
							\tau\int_\R(K-S_0e^x)^+\,\nu(\textrm{d}x),&S_0>K,\\
			                        \end{cases} 
\end{equation*}
as $\tau\to0^+$, i.e. as time to expiry goes to zero. (We have written $\nu$ for the L{\'e}vy measure of the driving L{\'e}vy noise.) In ``complete generality'', however, we can say only that
$$
	\Exp{(S_\tau-K)^+}-(S_0-K)^+=\Oh{\tau}
$$
as $\tau\to0^+$.

Using our results on the behaviour of call options close to expiry we show that implied volatility explodes as $\tau\to0^+$ in ``most'' exponential L{\' e}vy models.

Attention is restricted to calls and implied volatilities that are not at-the-money, i.e. $S_0\neq K$. 
\end{abstract}

\section{Introduction}
We model the stock as the exponential of a L{\' e}vy process, \emph{viz.}
$$
	S_\tau=S_0e^{X_\tau},\quad \forall\tau\geq0,
$$
where $X$ is a L{\' e}vy process and $S_0>0$ is some constant. For simplicity, we assume zero interest rate and dividend yield. The stock is a martingale.
The model is presented under the pricing measure, $\P$, chosen by the market so that we price options as expectations under $\P$ of
their payoff. Since $S$ is a time-homogeneous Markov process there is no loss of generality in assuming that we are at time zero. 

We show that in \emph{all} such models\footnote{We use the Landau notation. Let $a\in[-\infty,\infty]$. Then $f(x)\sim g(x)$ means $f(x)/g(x)\to1$ as $x\to a$, $f(x)=\Oh{g(x)}$ as $x\to a$ means that there is a neighbourhood of $a$ and a constant $\alpha>0$ such that $\abs{f(x)}\leq \alpha \abs{g(x)}$ in this neighbourhood. Finally, $f(x)=\oh{g(x)}$ when $f(x)/g(x)\to0$ as $x\to a$. Obvious adjustments are made for one-side limits.}
\begin{equation}\label{eq:OhCallasymptotics}
	\Exp{(S_\tau - K)^+} - (S_0-K)^+ = \Oh{\tau},\quad \textrm{ as }\tau\to0^+.
\end{equation}

In most cases of interest, we are able to identify the precise asymptotics as 
\begin{equation}\label{eq:ExactCallasymptotics}
	\Exp{(S_\tau-K)^+} - (S_0-K)^+ \sim 	\begin{cases}
							\tau\int_\R(S_0e^x-K)^+\,\nu(\textrm{d}x),&S_0<K,\\
							\tau\int_\R(K-S_0e^x)^+\,\nu(\textrm{d}x),&S_0>K,\\
			                        \end{cases} 
\end{equation}
where $\nu$ is the L{\' e}vy measure of $X$. 
Given that \eqref{eq:ExactCallasymptotics} holds, we show that the implied volatility of the $K$ strike call goes to $+\infty$ as maturity goes to zero. We also
show that implied volatility is always finite (strictly) prior to expiry. Some other cases are also considered.

We now review previous work on small expiry asymptotics for implied volatility and call option prices.
\subsection{Call option asymptotics}
In \cite{LevendorskiiPub}, Levendorski{\v i} calculates small expiry asymptotics for the European put and call in exponential L{\' e}vy models.
The results require some regularity conditions on the driving L{\' e}vy process that we are able to do without. He has also presented results in \cite{LevendorskiiLevy}, 
where he conjectured that Equation \eqref{eq:ExactCallasymptotics} holds in more generality.

In \cite{CarrWu}, Carr and Wu argue that\footnote{We will use $\widehat{S}$ for the stock price process when it is not assumed that it is an exponential of a L{\'e}vy process.} 
$$
	\Exp{(\widehat{S}_{t+\tau}-K)^+|\fancy{F}_t}-(\widehat{S}_t-K)^+= \Oh{\tau},\quad \textrm{ as }\tau\to0^+
$$
in the case that $\widehat{S}_t\neq K$. They deal with a more general class of processes that includes stochastic volatility. Their argument is very different 
to the one presented here. Not all exponential L{\' e}vy models are included in their setup: they require the L{\' e}vy measure to have a density.
We do not require this. It does not seem possible to obtain exact asymptotics using their method.

\subsection{Small time asymptotics of expectations in L{\'e}vy models}
In \cite{FigLopez}, Figueroa-L{\'o}pez obtains that
\begin{equation}\label{eq:FigLopezLimit}
	\lim_{\tau\to0^+}\dfrac{1}{\tau}\,\Exp{f(X_\tau)}=\int_\R f(x)\,\nu({\rm d}x),
\end{equation}
where $X$ is a L{\'e}vy process with L{\'e}vy measure $\nu$ and $f$ satisfies certain constraints on its behaviour as $\abs{x}\to0$ and $\abs{x}\to\infty$.
See also Jacod (\cite{Jacod}), Woerner (\cite{Woerner}), and  R{\"u}schendorf and Woerner (\cite{RuschendorfWoerner}). There is also a 
well-known result stated as Corollary 8.9 in Sato (\cite{Sato}) that \eqref{eq:FigLopezLimit} holds whenever $f$ is continuous and bounded and
vanishes in a neighbourhood of the origin. 

\subsection{Implied volatility asymptotics}
The best known works on implied volatilities near expiry are by Berestycki \emph{et al.} (\cite{BBF1} and \cite{BBF2}). The analysis is for 
diffusion models of the stock and all strikes. The class of models considered does not allow for implied volatility explosions. The behaviour of at-the-money implied 
volatility in stochastic volatility models with bounded volatility and a finite variation jump component has been considered by Durrleman (\cite{DurrlemanATM}). No 
implied volatility explosion occurs in the considered model and, in fact, the limiting value of the implied volatility turns out to be the
``instantaneous spot volatility'' of the model. We have shown, see \cite{GoldysRoper}, that the boundedness assumption on the volatility can be considerably relaxed.
In \cite{CarrWu}, Carr and Wu argue that implied volatility explodes as time to expiry goes to zero once jumps are included in their stochastic volatility with
jumps model. To the best of our knowledge there has been no rigorous study of the behaviour of implied volatility as expiry goes to zero in models including jumps.

\section{Setup and definitions}
For simplicity, we assume that interest rates and dividends are identically zero. 

We let the stock price $S$ be a non-negative martingale on a filtered probability space $(\Omega,\fancy{F},\F,\P)$ where $\F=(\fancy{F}_t)_{t\geq0}$ is a filtration satisfying the usual conditions. We assume that the process $S$ is such that
\begin{equation}\label{eq:SeqXexp}
	S_\tau=S_0e^{X_\tau},\quad \tau\geq0,
\end{equation}
where $S_0$ is a strictly positive constant and $X$ is a L{\'e}vy process.

We take $\P$ to be the pricing measure and so the price of a European call expiring at 
time $\tau\geq 0$ with strike $K(>0)$ is given by
$$
	\Exp{(S_{\tau}-K)^+}.
$$

We now recall the 
\begin{Definition}[L{\'e}vy process]
 A real-valued, $\F$-adapted process $X$ is a L{\'e}vy process provided that
\begin{enumerate}
 	\item $X_0=0$, $\P$-a.s.;
	\item $X$ is c{\`a}dl{\`a}g; 
	\item $X$ has stationary and independent increments; and
	\item $X$ is continuous in probability.
\end{enumerate}
\end{Definition}
We recall that a L{\' e}vy process is described by its characteristics triplet $(\gamma,\sigma,\nu)$ where
$\sigma,\gamma\in\R$, $\sigma\geq0$, and $\nu$ is a non-negative Radon measure satisfying
\begin{equation*}
\nu(\curlybr{0})=0\quad\textrm{ and }\int_\R (1\wedge y^2) \nu(\mydiff{y})<\infty.
\end{equation*}

So that our stock price process is a martingale we also assume that
\begin{equation}\label{eq:martconstraints}
\int_{\abs{y}>1}e^y\,\nu({\rm d}y)<\infty\textrm{ and }\gamma=-\dfrac{\sigma^2}{2}-\int_{\R}(e^y-1-y\mathbbm{1}_{\abs{y}\leq 1})\,\nu({\rm d}y),
\end{equation}
see Cont and Tankov (\cite{ContTankov}).
Let us choose now and fix for the remainder of this note $(\gamma,\sigma,\nu)$ and $S_0>0$.

\begin{Definition}
 A L{\'e}vy process $X$ is called trivial iff it has characteristic triplet $(\gamma,0,0)$. In our setup $S$ is a martingale so that $\gamma$ must be zero.
See Sato (p. 71 of \cite{Sato}).
\end{Definition}

To make the application of Figueroa-L{\'o}pez's result plain, we recall the definitions of submultiplicative and subadditive functions. 

\begin{Definition}
A non-negative locally bounded function $k:\R\to\R$ is submultiplicative if there exists a constant $\alpha>0$ such that
$k(x+y)\leq \alpha k(x)k(y)$ for all $x,y\in\R$.  
\end{Definition}

\begin{Definition}
A non-negative locally bounded function $q:\R\to\R$ is subadditive provided that there exists a constant $\beta>0$ such $q(x+y)\leq \beta (q(x)+q(y))$ for all $x,y\in\R$.  
\end{Definition}

We work with the following class of ``dominating functions''.
\begin{Definition}[The class $\mathcal{S}(\nu)$]
 A function $u:\R\to\R$ is of class $\mathcal{S}(\nu)$ if
\begin{enumerate}
 \item $u(x)=q(x)k(x)$ for some functions $q$ and $k$ where $q$ is subadditive and $k$ is submultiplicative; and
 \item $\int_{\abs{x}>1}\abs{u(x)}\nu(\textrm{d}x)<\infty$.
\end{enumerate}
\end{Definition}

We need the definition of implied volatility. 
\begin{Definition}[Implied volatility]
Let $\Sigma_t(K,\tau)$ denote the time $t$ implied volatility of a call option with strike $K>0$ expiring in $\tau$ time units. The stock $\widehat{S}$ is a 
positive martingale on an appropriate filtered probability space. Then $\Sigma_t(K,\tau)$ is the unique $\theta\in[0,\infty]$ that solves
$$
\widehat{S}_t\Phi\left(\dfrac{\ln(\widehat{S}_t/K)}{\theta\sqrt{\tau}}-\dfrac{\theta\sqrt{\tau}}{2}\right) - K\Phi\left(\dfrac{\ln(\widehat{S}_t/K)}{\theta\sqrt{\tau}}+\dfrac{\theta\sqrt{\tau}}{2}\right)=\Exp{(\widehat{S}_{t+\tau}-K)^+|\fancy{F}_t}.
$$
\end{Definition}

We will use the following ``put and call functions''.
\begin{Notation}\label{not:cp}
 For each fixed $K>0$, let $c:\R\to\R$ (respectively $p:\R\to\R$) be defined by $c(x)=(S_0e^x-K)^+$ (respectively $p(x)=(K-S_0e^x)^+$).
\end{Notation}

\section{Pertinent results from the literature}
\subsection{Results on L{\' e}vy processes}
The call option asymptotics we will obtain are a straightforward application of the following, rather recent, result.
\begin{Theorem}[Figueroa-L{\'o}pez, Theorem 1.1 in \cite{FigLopez}]\label{thm:FigLopez}
	Let $v:\R\to\R$ satisfy
\begin{enumerate}
	\item $v(x)=\oh{x^2}$ as $x\to0$;
	\item $v$ is locally bounded;
	\item $v$ is $\nu$-a.e. continuous; and
	\item there exists a function $u\in\mathcal{S}(\nu)$ for which
$$
	\limsup_{\abs{x}\to\infty}\dfrac{\abs{v(x)}}{u(x)}<\infty.
$$
\end{enumerate}
	Then
$$
	\tau^{-1}\Exp{v(X_\tau)}\to \int_\R v(x)\,\nu({\rm d}x)\textrm{ as }\tau\to0^+.
$$
\end{Theorem}

\begin{Remark}
 The theorem as stated here is actually a particular case of Theorem 1.1 in \cite{FigLopez}, taking condition (1) of Condition 1.1 in \cite{FigLopez}.
 Indeed, the theorem allows for various behaviours of $v$ in a neighbourhood of the origin depending on the characteristic triplet of $X$.
\end{Remark}

\begin{Remark}
We remark again that the result is well-known for bounded continuous 
functions vanishing on a neighbourhood of the origin, see \cite{Sato}, Corollary 8.9.
\end{Remark}

\begin{Theorem}[Support of the L{\'e}vy process at fixed times, Theorem 24.10 in \cite{Sato}]\label{thm:LevySupport}
 Suppose that $X$ is a L{\'e}vy process of drift $\gamma$ which is either of infinite activity and finite variation or of finite activity (and hence variation) with $0$ in the support of $\nu$. Then
\begin{enumerate}
 \item If the support of $\nu$ is a subset of $[0,\infty)$, then $\Prob{X_\tau\in [\gamma\tau,\infty)}=1$.
 \item If the support of $\nu$ is a subset of $(-\infty,0]$, then $\Prob{X_\tau\in (-\infty,\gamma\tau]}=1$.
\end{enumerate}
\end{Theorem}
\begin{Remark}
 This is a simplified version of Theorem 24.10 of \cite{Sato} that is sufficient for our purposes.
\end{Remark}

\subsection{Implied volatility asymptotics}
The implied volatility explosion we will demonstrate comes from Theorem \ref{thm:FigLopez} combined with the following
fact.
\begin{Theorem}[Roper and Rutkowski, Corollary 5.1 in \cite{me}]\label{thm:RoperRutkowski}
 Let $\Sigma_t(K,\tau)$ denote the time $t$ implied volatility of a European call with strike $K>0$ and expiry $t+\tau$.
 Let the stock price process $\widehat{S}$ be a positive (c{\`a}dl{\`a}g) martingale on an appropriate filtered probability space. Suppose that $K>0$ and $K\neq \widehat{S}_t$. If there exists a constant $\delta>0$, possibly depending on $\omega,K$ and $t$, such that, for every $\tau\in(0,\delta)$,
\begin{equation*}
  \Exp{(\widehat{S}_{t+\tau}-K)^+|\fancy{F}_t}>(\widehat{S}_t-K)^+,
\end{equation*}
then
  \begin{equation*}
\lim_{\tau\to0^+} \Sigma_t(K,\tau)=
\lim_{\tau\to0^+}\dfrac{\abs{\ln(K/\widehat{S}_t)}}{\sqrt{-2\tau\ln(\Exp{(\widehat{S}_{t+\tau}-K)^+|\fancy{F}_t}-(\widehat{S}_t-K)^+)}}
\end{equation*}
in the sense that the left-hand side limit exists (is infinite,
respectively) iff the right-hand side limit exists (is infinite, respectively)
and then they are equal. Also statements are considered in the $\P$-.a.s sense.
\end{Theorem}
\begin{proof}
In Corollary 5.1 of \cite{me}, take the call price to be the expectation of the call payoff. 
The three requisite assumptions on the call pricing function as well as the positivity of $\widehat{S}_t$
are then straightforward consequences of the assumption that $\widehat{S}$ is a positive 
(c{\`a}dl{\`a}g) martingale. 
\end{proof}
\begin{Remark}
 The behaviour of at-the-money implied volatilities in terms of the behaviour of the
at-the-money call is also described in \cite{me}. We do not use it here.
\end{Remark}

\section{Results}\label{sec:Results}
In this section we present our results. Proofs of the lemmas are relegated to the Appendix.

\subsection{Lemmas}
\begin{Lemma}\label{lem:eL2}
Let $S$ be defined by \eqref{eq:SeqXexp} and assume that $S$ satisfies the constraints set out in Equation \eqref{eq:martconstraints}.
For every $\tau>0$ and $K>0$, 
\begin{enumerate}[({A}i)]
	\item $\Exp{(S_\tau-K)^+}<S_0$.
	\item $\Exp{(S_\tau-K)^+}\geq (S_0-K)^+$.
	\item\label{itm:CattainsLBzero} Suppose that $S_0\leq K$, then $\Exp{(S_\tau-K)^+}=(S_0-K)^+$ iff $\Prob{S_\tau>K}=\Prob{X_\tau>\ln(K/S_0)}=0$.
	\item\label{itm:CattainsLBpos}  Suppose that $S_0> K$, then $\Exp{(S_\tau-K)^+}=(S_0-K)^+$ iff $\Prob{S_\tau<K}=\Prob{X_\tau<\ln(K/S_0)}=0$.
\end{enumerate}
In addition:
\begin{enumerate}[({B}i)]
	\item There exists a non-trivial L{\'e}vy process such that (A\ref{itm:CattainsLBzero}) holds for some $\tau,K>0$.
	\item There exists a non-trivial L{\'e}vy process such that (A\ref{itm:CattainsLBpos}) holds for some $\tau,K>0$.
\end{enumerate}
\end{Lemma}
\begin{proof}
 See the Appendix.
\end{proof}

\begin{Lemma}\label{lem:CallsAndPuts}
Let $S$ be defined by \eqref{eq:SeqXexp} and assume that $S$ satisfies the constraints set out in Equation \eqref{eq:martconstraints}.
Then Theorem \ref{thm:FigLopez} applies to $c(x)=(S_0e^x-K)^+$ (respectively $p(x)=(K-S_0e^x)^+$) provided that $K>S_0$ (respectively $K<S_0$) and $K>0$.
\end{Lemma}
\begin{proof}
 See the Appendix.
\end{proof}

\subsection{Call option asymptotics}
\begin{Theorem}[Small-expiry asymptotics of the European call option: General case]\label{prop:MostGeneralResult}

Suppose that $S_\tau=S_0e^{X_\tau}$, for all $\tau\geq0$, where $X$ is a L{\'e}vy process satisfying \eqref{eq:martconstraints} and $S_0>0$. Fix $K>0$ with $K\neq S_0$. Recall that $c(x)=(S_0e^x-K)^+$ and $p(x)=(K-S_0e^x)^+$. 

Then, for every $\tau,K>0$ where $K\neq S_0$,
$$
  \Exp{(S_{\tau}-K)^+} - (S_0-K)^+ = \Oh{\tau},\quad\tau\to0^+.
$$
\end{Theorem}
\begin{proof}
Recall the definition of $c$ and $p$ from Notation \ref{not:cp}. 
By Lemma \ref{lem:CallsAndPuts}, we may apply Theorem \ref{thm:FigLopez} to get that
\begin{equation}\label{eq:cplstar1}
	\lim_{\tau\to0^+} \dfrac{1}{\tau}\Exp{(S_{\tau}-K)^+}  = \int_\R c(x)\,\nu({\rm d}x),\quad (S_0<K),
\end{equation}
and
\begin{equation}\label{eq:cplstar2}
	\lim_{\tau\to0^+} \dfrac{1}{\tau}\Exp{(K-S_{\tau})^+}  = \int_\R p(x)\,\nu({\rm d}x),\quad (S_0>K).
\end{equation}
Now use put-call parity in \eqref{eq:cplstar2}. The claim is now clear: the right hand sides of \eqref{eq:cplstar1} and \eqref{eq:cplstar2}
lie in $[0,\infty)$. For example if $X$ has only positive jumps, then the right hand side of \eqref{eq:cplstar2} will be zero.
\end{proof}

\begin{Theorem}[Small-expiry asymptotics of the European call option]\label{prop:LessGeneralCallResult}

Suppose that $S_\tau=S_0e^{X_\tau}$, for all $\tau\geq0$, where $X$ is a L{\'e}vy process satisfying \eqref{eq:martconstraints} and $S_0>0$. 
Fix $K>0$. Recall that $c(x)=(S_0e^x-K)^+$ and $p(x)=(K-S_0e^x)^+$. 

Then, for all $S_0,K>0$ with $K\neq S_0$,
\begin{enumerate}[(i)]
	\item If $X$ has characteristic triplet $(0,0,0)$, then 
		$$
			\Exp{(S_\tau-K)^+}=(S_0-K)^+.
		$$
		for all $\tau\geq0$.
	\item If $X$ has characteristic triplet $(-\sigma^2/2,\sigma,0)$ with $\sigma\neq0$ then
		$$
			\Exp{(S_\tau-K)^+}=S_0\Phi\left(\dfrac{\ln(S_0/K)}{\sigma\sqrt{\tau}}-\dfrac{\sigma\sqrt{\tau}}{2}\right) - K\Phi\left(\dfrac{\ln(S_0/K)}{\sigma\sqrt{\tau}}+\dfrac{\sigma\sqrt{\tau}}{2}\right).
		$$
	\item\label{itm:S0ltK} If $S_0<K$, and $\int c(x) \nu({\rm d}x)>0$, then
		$$
			\Exp{(S_\tau-K)^+}-(S_0-K)^+\sim \tau \int c(x) \nu({\rm d}x),\quad\tau\to0^+.
		$$
	\item\label{itm:S0gtK} If $S_0>K$, and $\int_\R p(x) \nu({\rm d}x)>0$, then
		$$
			\Exp{(S_\tau-K)^+}-(S_0-K)^+\sim \tau \int_\R p(x) \nu({\rm d}x),\quad\tau\to0^+.
		$$
	\item Otherwise
		$$
			\Exp{(S_\tau-K)^+}-(S_0-K)^+ = \oh{\tau},\quad\tau\to0^+.
		$$
\end{enumerate}
\end{Theorem}
\begin{proof}\hspace*{\fill}
\begin{enumerate}[(i)]
	\item Trivial.
	\item This is nothing but the Black-Scholes model (with zero interest rates). See, for example, Musiela and Rutkowski (\cite{MusielaRutkowski}).
	\item See the proof of Theorem \ref{prop:MostGeneralResult}.
	\item See the proof of Theorem \ref{prop:MostGeneralResult}.
	\item Suppose that $S_0<K$. If $\int_\R c(x) \nu({\rm d}x)>0$, then (\ref{itm:S0ltK}) would apply. Suppose that $S_0>K$. If $\int_\R p(x) \nu({\rm d}x)>0$, then (\ref{itm:S0gtK}) would apply. Neither of these possibilities is the case. But, by Theorem \ref{prop:MostGeneralResult}, 
$$
  \Exp{(S_{\tau}-K)^+} - (S_0-K)^+ = \Oh{\tau},\quad\tau\to0^+,
$$
	and the claim follows.
\end{enumerate}
\end{proof}
\begin{Remark}
As noted by Levendorski{\v i} \cite{LevendorskiiLevy}, the asymptotics of the call depends only on the positive jumps of $X$ for
the case $S_0<K$ and the negative jumps of $X$ for the case $S_0>K$, and depends neither on the
drift nor the Gaussian part.
\end{Remark}

\subsection{Implied volatility}
\begin{Proposition}
  Suppose that $S_\tau=S_0e^{X_\tau}$, for all $\tau\geq0$, where $X$ is a L{\'e}vy process satisfying \eqref{eq:martconstraints} and $S_0>0$.
  For every $K,\tau>0$, the implied volatility of the European call for expiry $\tau$ and strike $K$, i.e. $\Sigma_0(K,\tau)$, satisfies
  \begin{equation*}
	0 \leq \Sigma_0(K,\tau) < \infty,
  \end{equation*}
 and the lower bound is obtained by trivial $X$ as well as for some non-trivial $X$.
\end{Proposition}
\begin{proof}
 Use Lemma \ref{lem:eL2} and the definition of implied volatility.
\end{proof}

\begin{Theorem}[Limiting implied volatility]\hspace*{\fill}

Suppose that $S_t=S_0e^{X_\tau}$, for all $\tau\geq0$, where $X$ is a L{\'e}vy process satisfying \eqref{eq:martconstraints} and $S_0>0$. We recall that $c(x)=(S_0e^x-K)^+$ and $p(x)=(K-S_0e^x)^+$

Then the implied volatility of a $K(>0),K\neq S_0$ strike European call satisfies the following. 
\begin{enumerate}[(i)]
	\item If $X$ has characteristic triplet $(0,0,0)$,  then for every $\tau,K>0$ $\Sigma_0(K,\tau)=0$ and so
	\begin{equation*}
			\lim_{\tau\to0^+}\Sigma_0(K,\tau)=0.
	\end{equation*}
	\item If $X$ has characteristic triplet $(-\sigma^2,\sigma,0)$ with $\sigma\in(0,\infty)$, then, for every $K,\tau>0$, $\Sigma_0(K,\tau)=\sigma$ and
	\begin{equation*}
		\lim_{\tau\to0^+}\Sigma_0(K,\tau)=\sigma 
	\end{equation*}
	\item If $S_0<K$, and $\int_\R c(x) \nu({\rm d}x)>0$, then
		\begin{equation*}
			\lim_{\tau\to0^+} \Sigma_0(K,\tau)=\infty
		\end{equation*}
	\item If $S_0>K$, and $\int_\R p(x) \nu({\rm d}x)>0$, then
		\begin{equation*}
			\lim_{\tau\to0^+} \Sigma_0(K,\tau)=\infty.
		\end{equation*}
	\item There exist non-trivial processes such that for some $K\neq S_0$
$$
		\lim_{\tau\to0^+}\Sigma_0(K,\tau)=0.
$$
\end{enumerate}
\end{Theorem}
\begin{proof}
The first two statements are trivial. The third and fourth follow from Theorem \ref{prop:LessGeneralCallResult} and Theorem \ref{thm:RoperRutkowski}
using the elementary fact that $\tau\ln(\tau)\to0$ as $\tau\to0^+$.

For the fifth we argue as follows.
Let $X_\tau$ be of infinite activity and finite variation with $0$ in the support of $\nu$. Suppose that the support of $\nu$ is a subset of $[0,\infty)$. With only
positive jumps $X$ must have negative drift in order that $S$ be a martingale. By Theorem \ref{thm:LevySupport}, we have that $\Prob{X_\tau\in[\gamma\tau,\infty)}=1$
so that $\Prob{X_\tau<\gamma\tau(<0)}=0$. Momentarily fix $\tau=1$. Choose $K$ such that $\ln(K/S_0)<\gamma\tau=\gamma(<0)$. Clearly,
$\Prob{X_1<\ln(K/S_0)(<\gamma\cdot1)}=0$. But then $\Prob{X_\tau<\ln(K/S_0)(<\gamma\tau)}=0$ for all $\tau\in(0,1)$. Using the same argument as in the proof of 
Lemma \ref{lem:eL2} we conclude that $\Exp{(S_\tau-K)^+}=(S_0-K)^+=S_0-K$ for all $\tau\in(0,1]$. It follows that $\Sigma_0(K,\tau)=0$ for all $\tau\in(0,1]$ so that, trivially, $\lim_{\tau\to0^+}\Sigma_0(K,\tau)=0$.
\end{proof}

\subsection{Examples}
If $X$ is a Variance Gamma, NIG, CGMY (KoBoL), Meixner, Generalized Hyperbolic, or Normal Tempered Stable process 
then the L{\'e}vy measure has a density that is positive at each point of $\R$ except zero, so 
\begin{equation*}
	\Exp{(S_\tau-K)^+} - (S_0-K)^+ \sim 	\begin{cases}
							\tau\int_\R c(x)\,\nu(\textrm{d}x),&S_0<K,\\
							\tau\int_\R p(x)\,\nu(\textrm{d}x),&S_0>K\\
			                        \end{cases} 
\end{equation*}
and so
\begin{equation*}
 \lim_{\tau\to0^+}\Sigma_0(K,\tau)=\infty.
\end{equation*}
Similarly, in Merton's model and Kou's model the L{\'e}vy measure has a density that is positive at each point of $\R$ except zero, so we get the implied volatility explosion. See Cont and Tankov (\cite{ContTankov}).

\section{Appendix}\label{sec:Appendix}
{\bf Lemma \ref{lem:eL2}\ }
For every $\tau>0$ and $K>0$, 
\begin{enumerate}[({A}i)]
	\item $\Exp{(S_\tau-K)^+}<S_0$.
	\item $\Exp{(S_\tau-K)^+}\geq (S_0-K)^+$.
	\item\label{itm:CattainsLBzeroAppdx} Suppose that $S_0\leq K$, then $\Exp{(S_\tau-K)^+}=(S_0-K)^+$ iff $\Prob{S_\tau>K}=\Prob{X_\tau>\ln(K/S_0)}=0$.
	\item\label{itm:CattainsLBposAppdx}  Suppose that $S_0> K$, then $\Exp{(S_\tau-K)^+}=(S_0-K)^+$ iff $\Prob{S_\tau<K}=\Prob{X_\tau<\ln(K/S_0)}=0$.
\end{enumerate}
In addition:
\begin{enumerate}[({B}i)]
	\item There exists a non-trivial L{\'e}vy process such that (A\ref{itm:CattainsLBzeroAppdx}) holds for some $\tau,K>0$.
	\item There exists a non-trivial L{\'e}vy process such that (A\ref{itm:CattainsLBposAppdx}) holds for some $\tau,K>0$.
\end{enumerate}
\begin{proof}
Fix $\tau,K>0$. 
\begin{enumerate}[({A}i)]
	\item	Clearly, $\Exp{(S_\tau-K)^+}\leq S_0$. Suppose that $\Exp{(S_{\tau}-K)^+}=S_0$. This
		requires that $\Exp{(S_{\tau}-K)^+}=\Exp{S_\tau}$. But, since $K>0$, $(S_{\tau}-K)^+<S_\tau$ $\P$-a.s., unless $S_\tau=0$, $\P$-a.s. However, $\Prob{X_\tau=-\infty}=0$.
	\item	Well known. Use Jensen's Inequality and then the martingale assumption on $S$.
	\item	Suppose that $S_0\leq K$. Then $\Exp{(S_{\tau}-K)^+}=(S_0-K)^+$, iff $\Prob{S_\tau>K}=0$ iff $\Prob{X_\tau>\ln(K/S_0)}=0$.
	\item	Suppose that $S_0>K$. Then 
		\begin{align*}
			&\textnormal{iff} & &\int_0^\infty(y-(y\wedge K))\,\Prob{S_\tau\in {\rm d}y}=\int_0^\infty y\,\Prob{S_\tau\in {\rm d}y}-K\\
			&\textnormal{iff} & &\int_0^\infty (y\wedge K)\,\Prob{S_\tau\in {\rm d}y}=\int_0^\infty K \,\Prob{S_\tau\in {\rm d}y}\\
			&\textnormal{iff} & &\int_{[0,K)} y\,\Prob{S_\tau\in {\rm d}y}=\int_{[0,K)}K\,\Prob{S_\tau\in {\rm d}y}\\
			&\textnormal{iff} & &\Prob{S_\tau<K}=0\\
			&\textnormal{iff} & &\Prob{X_\tau<\ln(K/S_0)}=0.
		\end{align*}
\end{enumerate}

\begin{enumerate}[({B}i)]
	\item 	Let $X$ be a finite variation, infinite activity L{\' e}vy process. Suppose first that the support of $\nu$ is a subset of 
		$[0,\infty)$. By Theorem \ref{thm:LevySupport}, we have $\Prob{X_\tau\in[\gamma \tau,\infty)}=1$. Because of the martingale property 
		of $S$ we must have $\gamma<0$.  We therefore take $K$ such that $\ln(K/S_0)<\gamma\tau(<0)$. Then $\Prob{X_\tau<\ln(K/S_0)}=0$.

	\item When $\nu$ is a subset of $(-\infty,0]$ we have, by Theorem \ref{thm:LevySupport}, that $\Prob{X_\tau\in(-\infty,\gamma\tau]}=1$. Because 
		of the martingale property of $S$ we must have $\gamma>0$. We therefore take $K$ such that $\ln(K/S_0)>\gamma\tau(>0)$. Then
		$\Prob{X_\tau>\ln(K/S_0)}=0$.
\end{enumerate}
\end{proof}

{\bf Lemma \ref{lem:CallsAndPuts}\ }
\emph{
For each fixed $K>0$, let $c:\R\to\R$ (respectively $p:\R\to\R$) be defined by $c(x)=(S_0e^x-K)^+$ (respectively $p(x)=(K-S_0e^x)^+$).
Then Theorem \ref{thm:FigLopez} applies to $c$ (respectively $p$) provided that $K>S_0$ (respectively $K< S_0$).}
\begin{proof}
Clearly $p$ and $c$ are locally bounded and $\nu$-a.e. continuous. 

For $c$, consider $x\mapsto e^x$.
By the constraint \eqref{eq:martconstraints} on $X$, the function $x\mapsto 1\cdot e^x$ is in $\mathcal{S}(\nu)$ and it satisfies 
$$
\limsup_{\abs{x}\to\infty}\dfrac{(S_0e^x-K)^+}{e^x}<\infty.
$$
For $K>S_0$, $c$ vanishes in a neighbourhood of the origin so that certainly $c(x)=\oh{x^2}$ as $x\to0$.

For $p$, consider $x\mapsto K$. Clearly $x\mapsto K\cdot 1$ is in $\mathcal{S}(\nu)$ and it satisfies
$$
	\limsup_{\abs{x}\to\infty}\dfrac{(K-S_0e^x)^+}{K}<\infty.
$$
For $K<S_0$, $p$ vanishes in a neighbourhood of the origin so that certainly $p(x)=\oh{x^2}$ as $x\to0$.
\end{proof}

\section{Acknowledgements}
Thanks are due to Ben Goldys and Marek Rutkowski for support, encouragement, and assistance. Thanks to Jan Baldeaux for
proof-reading. Remaining errors are my own.

\bibliographystyle{amsplain}
\providecommand{\bysame}{\leavevmode\hbox to3em{\hrulefill}\thinspace}
\providecommand{\MR}{\relax\ifhmode\unskip\space\fi MR }
\providecommand{\MRhref}[2]{%
  \href{http://www.ams.org/mathscinet-getitem?mr=#1}{#2}
}
\providecommand{\href}[2]{#2}

\end{document}